\documentclass[letterpaper, 10 pt, conference]{ieeeconf}  % Comment this line out if you need a4paper

\IEEEoverridecommandlockouts                              % This command is only needed if 
\title{\LARGE \bf
Gradient-Based STL Control with Application to Nonholonomic Systems
}

\author{Peter Varnai and Dimos V. Dimarogonas$^{1}$% <-this % stops a space
	\thanks{This work was partially supported by the Wallenberg AI, Autonomous Systems and Software Program (WASP) funded by the Knut and Alice Wallenberg Foundation, the Swedish Research Council (VR), the SSF COIN project, and the EU H2020 Co4Robots project.}% <-this % stops a space
	\thanks{$^{1}$Both authors are with the Division of Decision and Control Systems, School of Electrical Engineering and Computer Science, KTH Royal Institute of Technology, 114 28 Stockholm, Sweden. {\tt\small varnai@kth.se} (P. Varnai), {\tt\small dimos@kth.se} (D. V. Dimarogonas)}%
}

\usepackage{graphicx, subcaption}

\usepackage{amsmath, amssymb, amsthm}
\usepackage{accents}			% For a nice lower bar command \lbar

\usepackage{acro}
\usepackage{cite}

\usepackage{enumerate}
\usepackage{color}

\newtheoremstyle{bfplain}{}{}{\itshape}{}{\bfseries}{.}{ }{\thmname{#1}\thmnumber{ #2}\thmnote{ (#3)}}
\newtheoremstyle{bfdefinition}{}{}{}{}{\bfseries}{.}{ }{\thmname{#1}\thmnumber{ #2}\thmnote{ (#3)}}
\newtheoremstyle{itdefinition}{}{}{}{}{\itshape}{.}{ }{\thmname{#1}\thmnumber{ #2}\thmnote{ (#3)}}

\DeclareMathOperator*{\st}{s.t.}

\DeclareMathOperator*{\diag}{diag}

\newcommand{\satisfies}{\vDash}
\renewcommand{\and}{\wedge}

\NewDocumentCommand{\until}{oo}{
	\IfNoValueTF{#1}{\mathcal{U}}{\mathcal{U}_{[#1,#2]}}
}
\NewDocumentCommand{\eventually}{oo}{
	\IfNoValueTF{#1}{F}{F_{[#1,#2]}}
}
\NewDocumentCommand{\always}{oo}{
	\IfNoValueTF{#1}{G}{G_{[#1,#2]}}
}

\newcommand{\lbar}[1]{\underaccent{\bar}{#1}}

\newcommand{\tp}{^{\textsc{T}}}

\renewcommand{\iff}{\Leftrightarrow}

\newcommand{\bmat}[1]{\begin{bmatrix} #1 \end{bmatrix}}
\newcommand{\subnorm}[1]{_{#1}}
\NewDocumentCommand{\norm}{som}{
	\IfBooleanTF {#1}{\IfNoValueTF{#2} {\| #3 \|}{\| #3 \|\subnorm{#2}}}
	{\IfNoValueTF{#2} {	\left\| #3 \right\|}{\left\| #3 \right\|\subnorm{#2}}}
}
\NewDocumentCommand{\inReal}{soo}{
	\IfBooleanTF {#1}{}{\in}
	\mathbb{R}
	\IfNoValueTF{#2}{}{\IfNoValueTF{#3}{^{#2}}{^{#2 \times #3}}}	
}
\NewDocumentCommand{\inComplex}{soo}{
	\IfBooleanTF {#1}{}{\in}
	\mathbb{C}
	\IfNoValueTF{#2}{}{\IfNoValueTF{#3}{^{#2}}{^{#2 \times #3}}}	
}

\newcommand{\submat}[1]{_{\scriptstyle{#1}}}
\newcommand{\subvec}[1]{_{\scriptstyle{#1}}}
\newcommand{\supmat}[1]{^{\scriptstyle{#1}}}
\newcommand{\supvec}[1]{^{\scriptstyle{#1}}}

\newcommand{\mat}[1]{\mathbf{#1}}
\renewcommand{\vec}[1]{\boldsymbol{#1}}

\newcommand{\defvec}[2] {
	\DeclareDocumentCommand{#1}{s d<> d[] d''} {
		\IfBooleanTF {##1}
		{\IfNoValueTF{##2}{#2}{##2{#2}}}
		{\IfNoValueTF{##2}{\vec{#2}}{##2{\vec{#2}}}}		
		\IfNoValueTF{##3}{}{\subvec{##3}}
		\IfNoValueTF{##4}{}{\supvec{##4}}
	}
}

\newcommand{\defmat}[2] {
	\DeclareDocumentCommand{#1}{s d<> d[] d''} {
		\IfBooleanTF {##1}
		{\IfNoValueTF{##2}{#2}{[##2{#2}]}}
		{\IfNoValueTF{##2}{\mat{#2}}{##2{\mat{#2}}}}			
		\IfNoValueTF{##3}{}{\submat{##3}}
		\IfNoValueTF{##4}{}{\supmat{##4}}
	}
}
\defmat{\A}{A}
\defmat{\B}{B}
\defmat{\C}{C}
\defmat{\D}{D}
\defmat{\E}{E}
\defmat{\F}{F}
\defmat{\G}{G}
\defmat{\H}{H}
\defmat{\I}{I}
\defmat{\J}{J}
\defmat{\K}{K}
\defmat{\L}{L}
\defmat{\M}{M}
\defmat{\P}{P}
\defmat{\Q}{Q}
\defmat{\R}{R}
\defmat{\S}{S}
\defmat{\T}{T}
\defmat{\U}{U}
\defmat{\V}{V}
\defmat{\W}{W}
\defmat{\X}{X}
\defmat{\Y}{Y}
\defmat{\Z}{Z}
\defmat{\NMAT}{0}
\defmat{\SIG}{\Sigma}
\defmat{\LAM}{\Lambda}

\defvec{\ones}{1}
\defvec{\a}{a}
\defvec{\b}{b}
\defvec{\c}{c}
\defvec{\d}{d}
\defvec{\e}{e}
\defvec{\f}{f}
\defvec{\g}{g}
\defvec{\h}{h}
\defvec{\k}{k}
\defvec{\l}{l}
\defvec{\m}{m}
\defvec{\n}{n}
\defvec{\q}{q}
\defvec{\p}{p}
\defvec{\r}{r}
\defvec{\s}{s}
\defvec{\t}{t}
\defvec{\u}{u}
\defvec{\v}{v}
\defvec{\w}{w}
\defvec{\x}{x}
\defvec{\y}{y}
\defvec{\z}{z}
\defvec{\nvec}{0}
\defvec{\lam}{\lambda}
\defvec{\nuvec}{\nu}
\defvec{\pivec}{\pi}
\defvec{\phivec}{\phi}
\defvec{\rhovec}{\rho}
\defvec{\sigvec}{\sigma}
\defvec{\thetavec}{\theta}
\defvec{\alphavec}{\alpha}
\defvec{\gammavec}{\gamma}
\defvec{\Gammavec}{\Gamma}

\theoremstyle{bfdefinition}
\newtheorem{theorem}{Theorem}
\newtheorem{lemma}{Lemma}
\newtheorem{definition}{Definition}
\newtheorem{problem}{Problem}
\newtheorem{assumption}{Assumption}

\theoremstyle{itdefinition}
\newtheorem{remark}{Remark}
\newtheorem{example}{Example}
\newtheorem{corollary}{Corollary}[theorem]

\DeclareAcronym{TL}{
	short = TL,
	long = temporal logic
}
\DeclareAcronym{LTL}{
	short = LTL,
	long = linear temporal logic
}
\DeclareAcronym{TLTL}{
	short = TLTL,
	long = truncated linear temporal logic
}
\DeclareAcronym{STL}{
	short = STL,
	long = signal temporal logic
}
\DeclareAcronym{PPC}{
	short = PPC,
	long = prescribed performance control
}
\DeclareAcronym{PI2}{
	short = {PI$^2$},
	long = policy improvement with path integrals
}
\DeclareAcronym{TLPS}{
	short = TLPS,
	long = temporal logic policy search
}
\DeclareAcronym{ReLU}{
	short = ReLU,
	long = rectified linear unit
}
\DeclareAcronym{RL}{
	short = RL,
	long = reinforcement learning
}
\DeclareAcronym{MPC}{
	short = MPC,
	long = model predictive control
}
\DeclareAcronym{MDP}{
	short = MDP,
	long = Markov decision process
}
\DeclareAcronym{HJB}{
	short = HJB,
	long = Hamilton-Jacobi-Bellman
}

\begin{document}

\maketitle
\thispagestyle{empty}
\pagestyle{empty}

%% START OF OWN INCLUDES
%%%%%%%%%%%%%%%%%%%%%%%%%%%%%%%%%%%%%%%%%%%%%%%%%%%%%%%%%%%%%%%%%%%%%%%%%%%%%%%%
\begin{abstract}

In this paper, we study the control of dynamical systems under temporal logic task specifications using gradient-based methods relying on quantitative measures that express the extent to which the tasks are satisfied. A class of controllers capable of providing satisfaction guarantees for simple systems and specifications is introduced and then extended for the case of unicycle-like dynamics. The possibility of combining such controllers in order to tackle more complex task specifications while retaining their computational efficiency is examined, and the practicalities related to an effective combination are demonstrated through a simulation study. The introduced framework for controller design lays ground for future work in the direction of effectively combining such elementary controllers for the purpose of aiding exploration in learning algorithms.

\end{abstract}

%%%%%%%%%%%%%%%%%%%%%%%%%%%%%%%%%%%%%%%%%%%%%%%%%%%%%%%%%%%%%%%%%%%%%%%%%%%%%%%%
\section{INTRODUCTION}

In this work, we investigate control strategies for robotic systems subject to so-called \ac{TL} task specifications. \Aclp{TL} have many forms and allow an expression of rich and complex tasks through a combination of Boolean and temporal operators. Designing control strategies which guarantee that the system exhibits the desired behavior has gained considerable interest and is generally performed by abstracting the system space and applying solution techniques over such a discretized domain through high-level planning algorithms \cite{belta2017formal}.  

\Ac{STL} is a specific type of \ac{TL} which enables expressing tasks directly related to the system, without abstraction. The atomic predicates serving as a basis for these expressions are defined over functions of continuous-time system signals \cite{maler2004monitoring}. \ac{STL} allows placing temporal specifications on the evolution of these atomic predicates. This is useful in scenarios where explicit timing is important, such as having a robot visit a charging station within a fixed time span after its battery low indicator goes off. Previous works aim to provide controllers for solving \ac{STL} tasks using methods related to, e.g., \ac{MPC} \cite{raman2014model} or \ac{PPC} \cite{lindemann2017prescribed}. Reinforcement learning methods have also gained attention recently \cite{li2018policy} due to their success in other \ac{TL} languages \cite{fu2017sampling}.
	
Learning methods offer the possibility of dealing with unknown system dynamics as well as to potentially reuse gathered experience to tackle new tasks \cite{pan2010survey}. However, they rely on a multitude of simulations and experiments, which makes computational and sample efficiency crucial for their usability in practice, such as in the case of the policy improvement algorithm \cite{theodorou2010generalized}. Our work aims towards addressing this issue by presenting a framework for designing inexpensive, gradient-based controllers whose purpose is to guide exploration in such learning methods. Such guidance has been shown to yield significant improvements in the performance of policy improvement \cite{varnai2019prescribedARXIV, varnai2019learning}. The controllers sacrifice task satisfaction guarantees in exchange for computational efficiency as they are computed from an ensemble of elementary controllers related to simple subtasks.

The main contributions of the work presented in this paper are outlined as follows. First, a class of controllers with task satisfaction guarantees for simple tasks and dynamical systems is introduced. These controllers stem from prescribing the evolution of a task satisfaction metric in time, based on ideas from \ac{PPC} \cite{bechlioulis2008robust} as in \cite{lindemann2017prescribed}. The introduced framework is then used to extend the range of system dynamics which can be handled to unicycle-like models. Finally, we lay out initial thoughts regarding how to combine the derived controllers, e.g., to aid exploration while learning to solve complex tasks.

The paper is organized as follows. Section \ref{section:preliminaries} introduces \ac{STL} and the dynamical systems and task specifications under consideration. Section \ref{section:framework} derives a framework for gradient-based controller design for \ac{STL} specifications for simple systems. This is expanded to allow control of unicycle-like dynamics for specific forms of task specifications in Section \ref{section:extension}. Section \ref{section:practical} then discusses combining controllers from different task specifications and presents a related simulation study. Concluding remarks are given in Section \ref{section:conclusions}.

%%%%%%%%%%%%%%%%%%%%%%%%%%%%%%%%%%%%%%%%%%%%%%%%%%%%%%%%%%%%%%%%%%%%%%%%%%%%%%%%
\section{Preliminaries} \label{section:preliminaries}

\subsection{\Acf{STL}}

\ac{STL} is a type of predicate logic defined over continuous-time signals \cite{maler2004monitoring}. The \textit{predicates} $\mu$ are either true($\top$) or false($\bot$) according to the sign of a function $h^{\mu}:\inReal*[n] \rightarrow \inReal*$:
\begin{equation*}
\mu := \begin{cases}
	\top \text{ if } h^{\mu}(\x) \geq 0, \\
	\bot \text{ if } h^{\mu}(\x) < 0.
\end{cases}
\end{equation*}
Predicates are recursively combined using Boolean and temporal operators to form more complex \textit{task specifications} $\phi$:
\begin{equation*}
\phi := \top \ |\  \mu \ |\ \neg \phi \ |\ \phi_1 \and \phi_2 \ |\ \phi_1 \until[a][b]\phi_2,
\end{equation*}
where time bounds of the \textit{until} operator $\until[a][b]$ satisfy $a,b \in [0,\infty)$ as well as $a \le b$. The temporal operators \textit{eventually} and \textit{always} are defined from these by $\eventually[a][b]\phi = \top \until[a][b]\phi$ and $\always[a][b]\phi = \neg \eventually[a][b] \neg \phi$. A signal $\x(t)$ satisfies an \ac{STL} expression at time $t$ by the following semantics \cite{lindemann2017prescribed}:
\begin{align*}
&(\x, t) \satisfies \mu &&\iff h^{\mu}(\x(t)) \ge 0, \\
&(\x, t) \satisfies \neg\phi &&\iff \neg((\x, t) \satisfies \phi), \\
&(\x, t) \satisfies \phi_1 \and \phi_2 &&\iff (\x, t) \satisfies \phi_1 \and (\x, t) \satisfies \phi_2, \\
&(\x, t) \satisfies  \phi_1 \until{a}{b}\phi_2 &&\iff \exists t_1 \in [t+a, t+b] \ : \ (\x, t_1) \satisfies \phi_2\\
& && \quad \ \  \mathrm{and}\ (\x, t_2) \satisfies \phi_1 \  \forall t_2 \in [t, t_1],
\end{align*}
where the symbol $\satisfies$ denotes satisfaction of an \ac{STL} formula. 

Various robustness measures $\rho^{\phi}$ that quantify the extent to which a task specification $\phi$ is satisfied are summarized in \cite{donze2010robust}. In this work, we use the so-called \textit{spatial robustness} metric. For the types of tasks encountered in the presented case study example, this is evaluated recursively by: 
\begin{align*} 
\rho^\mu(\x, t) &= h^{\mu}(\x(t)) \\
\rho^{\neg \phi}(\x, t) &= -\rho^{\phi}(\x,t) \\
\rho^{\phi_1 \and \phi_2}(\x, t) &= \min\left(\rho^{\phi_1}(\x, t),\rho^{\phi_2}(\x, t)\right) \\
\rho^{\eventually[a][b]\phi}(\x, t)  &= \max_{t' \in [t+a,t+b]}\rho^{\phi}(\x,t') \\
\rho^{\always[a][b]\phi}(\x, t)  &= \min_{t' \in [t+a,t+b]}\rho^{\phi}(\x,t').
\end{align*}
A task is satisfied if its robustness metric is positive.

\subsection{System description}

Let us consider a nonlinear system of the form
\begin{equation} \label{eq:system}
\dot{\x} = f(\x) + g(\x) \u + \w, \qquad \x(0) = \x[0]
\end{equation}
with state $\x \inReal[n]$, input $\u \inReal[m]$, bounded process noise $\w \in \mathcal{B} \subset \mathbb{R}^n$, and initial state $\x[0] \inReal[n]$. The system is subject to some \ac{STL} task $\phi$ that is obtained by placing temporal specifications on a non-temporal formula $\psi$ composed of atomic predicates $\mu$ as follows:
\begin{equation*}
	\psi := \top\ |\ \mu \ | \ \neg \mu \ | \ \psi_1 \wedge \psi_2.
\end{equation*}
We assume that the temporal task $\phi$ is such that it can be satisfied by properly controlling the evolution of the robustness measure $\rho^{\psi}(\x)$ associated with $\psi$ in time; e.g., $\phi = \eventually[3][6]\psi$ requires $\rho^{\psi}(\x(t')) \ge 0$ for some $t' \in [3, 6]$. For a formal presentation and examples, see \cite{lindemann2017prescribed, varnai2019learning}. This assumption is stated as part of the following general assumptions.

\begin{assumption}[General assumptions] \label{assumption:general}
	The system and task definition are such that:
	\begin{enumerate}[(i)]
		\setlength{\itemsep}{0pt}
		\item the functions $f(\x)$, $g(\x)$, $\rho^{\psi}(\x)$ and its gradient $\frac{\partial \rho^{\psi}(\x)}{\partial \x}$  are locally Lipschitz continuous,
		\item the noise $\w(t)$ is piecewise continuous,
		\item there is a designed smooth curve $\gamma(t)$ such that $\rho^{\psi}(\x(t)) \ge \gamma(t)$ for all $t$ guarantees satisfaction of $\phi$, and
		\item the initial state $\x[0]$ is such that $\rho^{\psi}(\x[0]) \ge \gamma(0)$.		
	\end{enumerate}
\end{assumption}

The goal of the coming sections is to design a control law $\u(\x, t)$ which guarantees that the system satisfies the given task $\phi$, i.e., that the robustness specification $\rho^{\psi}(\x(t)) \ge \gamma(t)$ holds for all $t \ge 0$. The introduced mathematical derivations are primarily based on the following theorems.

\begin{lemma}[Theorem 3.1, Local Existence \& Uniqueness \cite{khalil2002nonlinear}] \label{lemma:odeSolution}
	Consider the initial value problem $\x<\dot> = f(\x, t)$ with given $\x(t_0) = \x[0]$. Suppose $f$ is uniformly Lipschitz continuous in $\x$ and piecewise continuous in $t$ in a closed ball $\mathcal{B} = \left\{\x \inReal[n], t \inReal : \norm{\x - \x[0]} \le r,\ t \in [t_0,\ t_1]\right\}$. Then, there exists some $\delta > 0$ such that the initial problem has a unique solution over the time interval $[t_0, t_0 + \delta]$.
\end{lemma}

\begin{lemma}[Theorem 3.3, \cite{khalil2002nonlinear}] \label{lemma:globalSolution}
	Consider the initial value problem of Lemma \ref{lemma:odeSolution}, where $f$ is piecewise continuous in $t$ and locally Lipschitz in $\x$ for all $t \ge t_0$ and all $\x$ in a domain $\mathcal{D} \subset \inReal*[n]$. If every solution of the system lies in a compact subset $\mathcal{W}$ of $\mathcal{D}$, then a unique solution exists to the initial value problem for all $t \ge t_0$.
\end{lemma}

\begin{lemma}[Generalized Nagumo's Theorem, {\cite[Section 4.2.2]{blanchini2015set}}] \label{lemma:nagumo}
	Consider the system $\x<\dot> = f(\x, t)$ and time-varying sets of the form $S(t) = \left\{\x : \zeta(\x, t) \le 0 \right\}$ where $\zeta(\x, t)$ is smooth. Assume that the system admits a unique solution and that at any $t$ we have $\frac{\partial \zeta(\x, t)}{\partial \x} \ne \nvec$ for $\zeta(\x, t) = 0$. The condition $x(\tau) \in S(\tau)$ implies $x(t) \in S(t)$ for $t \ge \tau$ if the inequality $\dot{\zeta}(\x, t) \le 0$ holds at the boundary $\zeta(\x, t) = 0$.
%	\begin{equation}
%		\dfrac{\partial \zeta(\x, t)}{\partial \x} f(t, \x) + \dfrac{\partial \zeta(\x, t)}{\partial t} \le 0.
%	\end{equation}
\end{lemma}

%%%%%%%%%%%%%%%%%%%%%%%%%%%%%%%%%%%%%%%%%%%%%%%%%%%%%%%%%%%%%%%%%%%%%%%%%%%%%%%%
\section{Gradient-based \ac{STL} control framework} \label{section:framework}

This section presents a framework for different gradient-based control approaches to solving \ac{STL} tasks, relating to earlier work using the \acs{PPC} and barrier function methods \cite{lindemann2017prescribed, lindemann2019control}. Intuitively, the system (\ref{eq:system}) only needs to be controlled when the robustness measure nears the specification curve $\gamma(t)$ in order to guarantee the desired $\rho^{\psi}(\x(t)) \ge \gamma(t)$. This motivates the following definitions.
\begin{definition}[Region of interest] \label{def:roi}
	Let $\varGamma(t)$ be a smooth curve for which $\varGamma(t) \ge \gamma(t) + \epsilon$ for all $t \ge 0$ and some $\epsilon > 0$. The \textit{region of interest} $\mathcal{X}(t)$ at time $t$ is defined as:
	\begin{equation}
	\mathcal{X}(t) := \left\{\x \inReal[n] : \gamma(t) \le \rho^{\psi}(\x) \le \varGamma(t)\right\}.
	\end{equation}
	The upper and lower boundaries of this region are denoted by the two sets $\bar{\mathcal{X}}(t) := \left\{\x \inReal[n] : \rho^{\psi}(\x) = \varGamma(t)\right\}$ and $\lbar{\mathcal{X}}(t) := \left\{\x \inReal[n] : \rho^{\psi}(\x) = \gamma(t)\right\}$. We also introduce the \textit{uncontrolled region} $\mathcal{A}(t) := \left\{\x \inReal[n] : \rho^{\psi}(\x) > \varGamma(t)\right\}$.
\end{definition}

\begin{definition}[Local robustness satisfaction]
	Let the system (\ref{eq:system}) be controlled by $\u = \u(\x,t)$. This control law is said to locally satisfy the robustness specification $\rho^{\psi}(\x(t)) \ge \gamma(t)$ in a domain $\mathcal{D} \subseteq \inReal*[n]$ if, for any initial $\x(\tau) \in \mathcal{D}$ such that $\rho^{\psi}(\x(\tau)) \ge \gamma(\tau)$, there exists a time $\delta > 0$ for which $\rho^{\psi}(\x(t)) \ge \gamma(t)$ holds during the interval $t \in [\tau, \tau+\delta]$. 
\end{definition}

\subsection{General control law design}

Let us examine the temporal behavior of the robustness measure $\rho^{\psi}(\x)$ that is to be controlled for the system (\ref{eq:system}):
\begin{equation} \label{eq:rhoDot} \hspace{-1mm}
\dot{\rho}^{\psi}(\x) = \dfrac{\partial \rho^{\psi}(\x)}{\partial \x} \x<\dot> = \underbrace{\dfrac{\partial \rho^{\psi}(\x)}{\partial \x} \left(f(\x) + \w\right)}_{\dot{\rho}^{\psi}_{fw}(\x, \w)} + \underbrace{\dfrac{\partial \rho^{\psi}(\x)}{\partial \x} g(\x) \u}_{\dot{\rho}^{\psi}_{u}(\x)},
\end{equation}
where $\dot{\rho}^{\psi}_{u}(\x)$ denotes the term influenced by $\u$, as implied by the subscript.

For developing our framework, in this section we consider the case of simple system dynamics that essentially allow direct control over the evolution of $\rho^{\psi}(\x)$. To ease notation, define
\begin{equation} \label{eq:vDef}
	\v(\x)\tp := \dfrac{\partial \rho^{\psi}(\x)}{\partial \x} g(\x)
\end{equation}
by which we can simply express $\dot{\rho}^{\psi}_{u}(\x)$ as $\v(\x)\tp \u$.
\begin{assumption} \label{assumption:controllability}
	For the term $\v(\x)$, we have:
	\begin{equation} \label{eq:controllability}
		\v(\x) \ne \nvec, \quad \forall \x : \exists t\ \st \ \x \in \mathcal{X}(t).
	\end{equation}
\end{assumption}

\begin{remark}
	The derivations in \cite{lindemann2017prescribed} consider the assumptions $g(\x) g(\x)\tp > 0$, $\rho^{\psi}(\x)$ being concave with optimum $\rho^{\psi}_{\text{opt}}$, and $\varGamma(t) < \rho^{\psi}_{\text{opt}}$. These form a subset of Assumption \ref{assumption:controllability}. Since $g(\x) g(\x)\tp > 0$, $g(\x)$ is full row rank and thus $\v(\x)$ can become zero if and only if $\frac{\partial \rho^{\psi}(\x)}{\partial \x} = \nvec$. This gradient is non-zero for all $\x$ for which $\rho^{\psi}(\x) \ne \rho^{\psi}_{\text{opt}}$ as $\rho^{\psi}(\x)$ is concave. Thus, (\ref{eq:controllability}) holds for all $\x \in \mathcal{X}(t)$ as $\varGamma(t) < \rho^{\psi}_{\text{opt}}$.
\end{remark}

\begin{theorem} \label{theorem:localSatisfaction}
Let Assumptions \ref{assumption:general} and \ref{assumption:controllability} hold. Define
\begin{equation} \label{eq:baseControlFamily}
\u(\x,t) := \begin{cases}
	\nvec \qquad &\text{if } \x \in \mathcal{A}(t), \\
	\kappa(\x, t) \dfrac{K}{\norm[2]{\v(\x)}^{2} + \varDelta}\v(\x) &\text{if } \x \notin \mathcal{A}(t),
\end{cases}
\end{equation}
where the coefficient $\kappa(\x, t) \ge 0$ is continuous in $t$, locally Lipschitz in $\x$, and satisfies (i) $\kappa(\x, t) \ge \dot{\gamma}(t) + B(\x)$ with $B(\x) \ge -\dfrac{\partial \rho^{\psi}(\x)}{\partial \x} f(\x) + \max_{\w}\norm[2]{\dfrac{\partial \rho^{\psi}(\x)}{\partial \x} \w}$ for all $\x \in \lbar{\mathcal{X}}(t)$ and (ii) $\kappa(\x, t) = 0$ for all $\x \in \bar{\mathcal{X}}(t)$. Then, with a proper choice of the additional parameters $K \ge 1$ and $\varDelta \ge 0$, this control law achieves local robustness satisfaction of the specification $\rho^{\psi}(\x(t)) \ge \gamma(t)$ for the system (\ref{eq:system}) in the entire domain $\inReal*[n]$.
\end{theorem}

\begin{proof}
	Let the system at time $\tau$ be at a state $\x(\tau)$ for which $\rho^{\psi}(\x(\tau)) \ge \gamma(\tau)$. To prove local robustness satisfaction, we show that under the defined control law a unique solution exists for which $\rho^{\psi}(\x(t)) \ge \gamma(t)$ and remains satisfied for some period of time. For the former, in order to apply Lemma \ref{lemma:odeSolution}, we must show that there exists a closed ball around $\x(\tau)$ and $\tau$ within which $\u(\x, t)$ is Lipschitz continuous in $\x$ and piecewise continuous in $t$. Then the same holds for $f(\x) + g(\x)\u(\x, t) + \w(t)$, the right hand side of (\ref{eq:system}), due to Assumption \ref{assumption:general} (i) and (ii), and the lemma can be applied. 
	
	Piecewise continuity in $t$ trivially holds due to the continuity of $\kappa(\x, t)$ and $\mathcal{A}(t)$ in $t$. The Lipschitz condition also holds trivially for any $\x \in \mathcal{A}(t)$ where the control is defined to be zero. If $\x(\tau) \notin \mathcal{A}(t)$, then we must have $\x(\tau) \in \mathcal{X}(\tau)$ for which $\norm[2]{\v(\x(\tau))} \ge v_{\text{min}}$ for some $v_{\text{min}} > 0$ by the extreme value theorem and Assumption \ref{assumption:controllability}. Thus, as $\v(\x)$ is continuous, there exists a closed ball $\mathcal{B}$ around $\x(\tau)$ in which $\norm[2]{\v(\x)}$ is nonzero. Furthermore, as $\v(\x)$ and $\kappa(\x,t)$ are locally Lipschitz, the control action (\ref{eq:baseControlFamily}) is also Lipschitz in $\mathcal{B}$ (even in the case $\varDelta = 0$ as $\norm[2]{\v(\x)} \ne 0$). The Lipschitz property of $\u(\x, t)$ is preserved at the boundary $\bar{\mathcal{X}}(t)$ where $\u$ is continuous. Therefore, Lemma \ref{lemma:odeSolution} is applicable and a unique solution exists for some time interval $t \in [\tau, \tau + \delta]$ from the initial condition $\x(\tau)$.
	
	The proof of local robustness satisfaction is completed by showing that during this time $\rho^{\psi}(\x(t)) \ge \gamma(t)$ remains true (for any time interval, in fact, for which a solution exists). A sufficient condition for this is given by extensions of Nagumo's Theorem (see Lemma \ref{lemma:nagumo}). Applying the lemma to the set defined as $S(t) = \left\{\x: \gamma(t) - \rho^{\psi}(\x) \le 0 \right\}$ yields the condition: \vspace{-2mm}
	\begin{equation} \label{eq:satisfactionCondition}
	\dot{\rho}^{\psi}(\x(t)) \ge \dot{\gamma}(t) \quad \text{if } \x \in \lbar{\mathcal{X}}(t),
	\end{equation} 
	which, if satisfied, implies that the trajectory of $\rho^{\psi}(\x(t))$, having started above $\gamma(t)$, cannot cross it, as desired. Let the controller parameters satisfy $(K - 1)v^2_{\text{min}} \ge \varDelta$, e.g., with $K = 1$ and $\varDelta = 0$. Then, as $\norm[2]{v(\x)} \ge v_{\text{min}}$, we also have $(K - 1)\norm[2]{v(\x)}^2 \ge \varDelta$ for all $\x \in \mathcal{X}(t)$, thus the inequality
	\begin{equation}
	\dfrac{K}{\norm[2]{v(\x)}^2 + \varDelta} \ge \dfrac{1}{\norm[2]{v(\x)}^2}
	\end{equation}
	holds in this set as well. Substituting the control law (\ref{eq:baseControlFamily}) at $\x \in \lbar{\mathcal{X}}(t)$ into the time derivative (\ref{eq:rhoDot}) of $\rho^{\psi}$, and using the imposed bounds on $\kappa(\x, t)$, we can show that Nagumo's condition is then satisfied at the required $\x \in \lbar{\mathcal{X}}(t)$ region:
	\begingroup
	\allowdisplaybreaks
	\begin{align*}
		\dot{\rho}^{\psi}(\x) &= \dfrac{\partial \rho^{\psi}(\x)}{\partial \x} \left(f(\x) + \w\right) + \v(\x)\tp \dfrac{\kappa(\x, t)K}{\norm[2]{\v(\x)}^{2} + \varDelta}\v(\x) \\
		&\ge \dfrac{\partial \rho^{\psi}(\x)}{\partial \x} \left(f(\x) + \w\right) + \dfrac{\kappa(\x, t)}{\norm[2]{\v(\x)}^{2}}\v(\x)\tp \v(\x) \\
		&\ge \dfrac{\partial \rho^{\psi}(\x)}{\partial \x} f(\x) - \max_{\w}\norm{\dfrac{\partial \rho^{\psi}(\x)}{\partial \x} \w} \\
		&\phantom{\ge}+ \dot{\gamma}(t) - \dfrac{\partial \rho^{\psi}(\x)}{\partial \x} f(\x) + \max_{\w}\norm{\dfrac{\partial \rho^{\psi}(\x)}{\partial \x} \w} \\
		&= \dot{\gamma}(t),
	\end{align*}
	\endgroup
	as was to be shown for local robustness satisfaction.
\end{proof}

\begin{theorem} \label{theorem:globalSatisfaction}
	Assume the evolution of the system (\ref{eq:system}) under a locally robustness satisfying control law is such that the state remains bounded. Then, under Assumption \ref{assumption:general}, the corresponding \ac{STL} task $\phi$ is also satisfied.
\end{theorem}

\begin{proof}
	If the state remains bounded, a solution must exist for the entire time duration $t \ge t_0$ by Lemma \ref{lemma:globalSolution}. As the initial condition satisfies $\rho^\psi(\x(t_0)) \ge \gamma(t_0)$, by definition $\rho^\psi(\x(t)) \ge \gamma(t)$ must remain true for all $t \ge t_0$ since the control law is locally robustness satisfying. This in turn implies satisfaction of the task $\phi$ due to the design of the specification curve $\gamma(t)$. 
\end{proof}

Note that we do not require the controller (\ref{eq:baseControlFamily}) to guarantee the existence of a solution for all $t \ge t_0$. Indeed, suppose a robot needs to avoid collision with a stationary obstacle. This can be accomplished by using a locally robustness satisfying controller whose region of interest consists of points near the obstacle. Outside this region (in $\mathcal{A}(t)$), the controller allows the robot to evolve under its autonomous dynamics, where the system might have finite escape time. This choice is motivated by how we will aim to combine controllers from various robustness specifications. These would interfere more with each other if they were aiming to maintain a system solution outside their respective regions of interest. Keeping the state bounded to guarantee the existence of a global solution can simply be viewed as an added task specification.

\begin{corollary} \label{corollary:disjointConjunctions}
	Consider the conjunction of $M$ specifications $\rho^{\psi_{(i)}}(\x(t)) \ge \gamma_{(i)}(t)$ whose overall local robustness satisfaction guarantees that the system state remains bounded. Furthermore, assume that the specification curves $\gamma_{(i)}(t)$ and $\varGamma_{(i)}(t)$ are such that their defined regions of interest are mutually disjoint, i.e. $\mathcal{X}_{(i)} \cap \mathcal{X}_{(j)} = \emptyset$ for any $i,j \in {1,\dots,M}, i \ne j$. Then, for any control laws $\u[(i)](\x, t)$ that achieve local robustness satisfaction of the individual specifications, i.e., $\rho^{\psi_{(i)}}(\x(t)) \ge \gamma_{(i)}(t)$, the overall control $\u(\x, t) = \sum_{i=1}^{M} \u[(i)](\x, t)$ guarantees global robustness satisfaction of their conjunction.
\end{corollary}

%\begin{proof}
%	Since each $\mathcal{X}_i$ are disjoint, taking the sum of the individual $\u_i$ controls simply means that at the respective $i$th region of interest $\mathcal{X}_i$ we use the $\u[i]$ control action (the others are zero at these states). Since each $\u[i]$ is locally robustness satisfying, $\u$ is locally robustness satisfying for all $\rho^{\psi_i}(\x(t)) \ge \gamma_i(t)$ specifications. By the assumption of the corollary, such local robustness satisfaction implies that the state remains bounded, which in turn implies by Theorem \ref{theorem:globalSatisfaction} that a solution exists as $t \rightarrow \infty$. Thus, the conjunction $\phi = \bigcap \phi_i$ is globally satisfied as well due to the design of the performance curves $\gamma_i(t)$.
%\end{proof}

\begin{proof}
	The corollary follows directly from the independent regions of interest for the individual $\u_{(i)}$ control actions (i.e., at any time only a single one of them is nonzero) and the results of Theorems \ref{theorem:localSatisfaction} and \ref{theorem:globalSatisfaction}.
\end{proof}

\begin{remark}
	If the conjoined satisfaction of the $M$ specifications guarantees that the system state remains in some $\mathcal{D}$ domain, then Assumptions \ref{assumption:general} (i) and \ref{assumption:controllability} can be relaxed to hold for only the states $\x \in \mathcal{D}$.
\end{remark}

\begin{remark}
	Equation (\ref{eq:baseControlFamily}) defines a family of controllers based on the controller parameter $\kappa(\x, t)$. The choice $\kappa(\x,t) \rightarrow \infty$ as $\x \rightarrow \lbar{\mathcal{X}}(t)$ leads to an aggressive controller used in \cite{lindemann2017prescribed} and allows task satisfaction even if the dynamics $f(\x)$ and noise $\w$ are unknown. On the other hand, satisfying $\kappa(\x, t) \ge \dot{\gamma}(t) + B(\x)$ at $\x \in \lbar{\mathcal{X}}(t)$ by an exact equality is minimally invasive, but assumes full knowledge of the system dynamics. This is similar to the barrier function method described in \cite{lindemann2019control}, which even allows controllers for combined robustness specifications in the form of a single barrier function. The trade-off there appears in the nontrivial design of barrier functions and the added expense of computing the control actions through quadratic optimization.
	
	Many controllers lie in between these two outlined extremes. For example, an estimate $\bar{B}$ of the upper bound of $(\dot{\gamma}(t) + B(\x))$ could lead to $\kappa(\x, t) := \bar{B} e^{-\frac{\rho^{\psi}(\x) - \gamma(t)}{\varGamma(t) - \rho^{\psi}(\x)}}$. The aggressiveness of controller actions is mitigated, and explicit knowledge of the system dynamics $f(\x)$ is not required; however, depending on the estimate $\bar{B}$, task satisfaction guarantees could be lost. Such controllers were also used in \cite{varnai2019prescribedARXIV} and can be expected to be better combined due to their mitigated aggressiveness, thus aiding exploration more effectively when solving more complex \ac{STL} tasks using learning methods. Section \ref{section:practical} gives practical insights into how control actions from various robustness specifications can be combined into a single control action.
\end{remark}

\section{Extension to unicycle-type dynamics} \label{section:extension}

Our goal is to use the developed framework to devise locally task satisfying controllers for a wider range of system dynamics. The following example illustrates how control failure can occur even in the simple case of unicycle dynamics, motivating the extension studied in this paper.

\begin{example}[Unicycle navigation task] \label{example:1}	
	Consider a unicycle with state $\x = [x\ y\ \theta]\tp$, input $\u = [v\ \omega]\tp$, and dynamics:
	\begin{equation}
		\dot{x} = v\cos\theta, \qquad \dot{y} = v\sin{\theta}, \qquad \dot{\theta} = \omega.
	\end{equation}
	Aiming to navigate within a distance $r_g$ of a given goal $[x_g \ y_g]\tp$, a non-temporal formula $\psi$ is defined by the robustness measure $\rho^{\psi}(\x) = r_g - \norm[2]{\e[g]}$, where the target error is $\e[g] = [x - x_g \ \ y - y_g]\tp$. A temporal task is imposed as $\phi = \eventually[0][10]G\psi$. This temporal behavior is guaranteed if $\rho^{\psi}(\x(t)) \ge \gamma(t)$ for a curve $\gamma(t)$ which remains non-negative after some $t' \in [0, 10]$, i.e., the unicycle eventually always stays in the target region. The term $\v(\x)$ given by (\ref{eq:vDef}) in this case takes the form:
	\begin{equation*}
		\v(\x) = -\dfrac{1}{\norm[2]{\e[g]}}\bmat{\n\tp & 0 \\ \nvec & 1} \cdot \bmat{\e[g] \\ 0} = -\dfrac{1}{\norm[2]{\e[g]}}\bmat{\n\tp \e[g] \\ 0},
	\end{equation*}
	where $\n\tp =  [\cos(\theta)\ \sin(\theta)]$ is the heading direction of the unicycle. The first element and thus $\v(\x)$ can be zero for any $x$ and $y$ in case the error vector is perpendicular to $\n$, i.e., even when $\x \in \mathcal{X}(t)$, violating Assumption \ref{assumption:controllability}. Such a configuration could be avoided by properly changing $\theta$ in time using the input $\omega$; however, the second element of $\v(\x)$ is zero, so the derived controller (\ref{eq:baseControlFamily}) would not do so.
\end{example}

The exemplified controller failure motivates the following problem statement discussed in this section.

\begin{problem} \label{problem1}
	Consider the nonlinear system (\ref{eq:system}) with the following specific form (that also encompasses the unicycle):
	\begin{equation} \label{eq:unicycleSystem}
		\x<\dot> := \bmat{\x<\dot>[1] \\ \x<\dot>[2]} = \bmat{f_{1}(\x[1]) \\ f_{2}(\x)} + \bmat{g_{11}(\x[2]) & \NMAT \\ g_{21}(\x) & g_{22}(\x)} \bmat{\u[1] \\ \u[2]} + \bmat{\w[1] \\ \w[2]}.
	\end{equation}
	Determine a domain $\mathcal{D}$ and assumptions necessary for the local robustness satisfaction of $\rho^{\psi}(\x) \ge \gamma(t)$, and design a control law which achieves this, in case $\rho^{\psi}(\x)$ only depends on the state $\x[1]$ and with a slight abuse of notation can be written as $\rho^{\psi}(\x[1])$.
\end{problem}

%%%%%%%%%%%%%%%%%%%%%%%%%%%%%%%%%%%%%%%%%%%%%%%%%%%%%%%%%%%%%%%%%%%%%%%%%%%%%%%%
\subsection{Controller design} \label{section:solution}

To begin our study of Problem \ref{problem1}, let us express the time derivative of the robustness metric $\rho^{\psi}$ for the system (\ref{eq:unicycleSystem}). 
\begin{equation} \label{eq:unicycledRho}
\dot{\rho}^{\psi}(\x[1]) = \dfrac{\partial \rho^{\psi}(\x[1])}{\partial \x[1]} \x<\dot>[1] = \dfrac{\partial \rho^{\psi}(\x[1])}{\partial \x[1]} (f_{1}(\x[1]) + \w[1]) + \v(\x)\tp \u[1],
\end{equation}
where the term $\v(\x)$ has been redefined following (\ref{eq:vDef}) as
\begin{equation} \label{eq:vDef2}
\v(\x)\tp := \dfrac{\partial \rho^{\psi}(\x[1])}{\partial \x[1]} g_{11}(\x[2]).
\end{equation}
The results for a controller of the form (\ref{eq:baseControlFamily}) are not applicable to calculate the control action $\u[1]$, because $\v(\x)$ may become zero in the region of interest $\mathcal{X}(t)$ of the robustness specification $\rho^{\psi}(\x[1](t)) \ge \gamma(t)$ (as highlighted by Example \ref{example:1} for the unicycle scenario).

The idea is to avoid $\v(\x) = \nvec$ using an augmented task $\phi_{\text{aug}} := G \psi_{\text{aug}}$, where $\psi_{\text{aug}}$ is the non-temporal specification of keeping $\v(\x)$ non-zero by some small predefined $v_{\text{min}} > 0$ value:
\begin{equation} \label{eq:augmentedSpec}
\rho^{\psi_{\text{aug}}}(\x) := \norm[2]{\v(\x)} - v_{\text{min}}.
\end{equation}
Suitable robustness specification curves for this \textit{always} type task could be the constant values $\gamma_{\text{aug}}(t) = 0$ and $\varGamma_{\text{aug}}(t) = \alpha > 0$ used herein. The augmented task is thus to keep $\rho^{\psi_{\text{aug}}}(\x(t)) \ge \gamma_{\text{aug}}(t)$. The region of interest defined by these curves according to Definition \ref{def:roi} is denoted by $\mathcal{X}_{\text{aug}}(t)$, i.e., $\mathcal{X}_{\text{aug}}(t) = \left\{\x \inReal[n] : \gamma_{\text{aug}}(t) \le \rho^{\psi_{\text{aug}}}(\x) \le \varGamma_{\text{aug}}(t)\right\}$. The quantities $\lbar{\mathcal{X}}_{\text{aug}}(t)$, $\bar{\mathcal{X}}_{\text{aug}}(t)$, and $\mathcal{A}_{\text{aug}}(t)$ follow Definition \ref{def:roi} as well. The notation for the prescribed curves $\gamma(t)$, $\varGamma(t)$, the region of interest $\mathcal{X}(t)$, and the uncontrolled region $\mathcal{A}(t)$ is kept in relation to the original formula $\psi$.

Note that if the conjoined specifications for $\psi$ and $\psi_{\text{aug}}$ are satisfied, then the system state is guaranteed to stay within $\mathcal{D} := \left\{\x : \exists t, \x \in (\mathcal{X} (t) \cup \mathcal{A} (t)) \cap (\mathcal{A}_{\text{aug}}(t) \cup \mathcal{X}_{\text{aug}}(t))\right\}$. By definition of $\mathcal{D}$ and the augmented task (\ref{eq:augmentedSpec}), we thus have:
\begin{equation} \label{eq:unicycleV}
\norm[2]{\v(\x)} \ge v_{\text{min}}, \quad \forall \x \in \mathcal{D} : \exists t\ \st\   \x \in \mathcal{X}(t).
\end{equation}

\begin{lemma} \label{lemma:phiSatisfaction}
	Assume that the control law $\u[2](\x, t)$ for input $\u[2]$ is Lipschitz continuous in $\x$ and piecewise continuous in $t$ in the region $\mathcal{D}$ and that Assumption \ref{assumption:general} holds. Define $\u[1]$ as:
	\begin{equation} \label{eq:unicycleControlu1}
	\u[1](\x,t) = \begin{cases}
	\nvec \qquad &\text{if } \x \in \mathcal{A}(t), \\
	 \dfrac{\kappa_1(\x[1], t) K}{\norm[2]{\v(\x)}^{2} + \varDelta}\v(\x) \qquad &\text{if } \x \notin \mathcal{A}(t),
	\end{cases}
	\end{equation}
	where the coefficient $\kappa_1(\x[1], t) \ge 0$ is continuous in $t$, locally Lipschitz in $\x$, and satisfies (i) $\kappa_1(\x[1], t) \ge \dot{\gamma}(t) + B_1(\x)$ with $B_1(\x) \ge -\dfrac{\partial \rho^{\psi}(\x[1])}{\partial \x[1]} f_{1}(\x[1]) + \max_{\w[1]}\norm[2]{\dfrac{\partial \rho^{\psi}(\x[1])}{\partial \x[1]} \w[1]}$ for all $\x \in \lbar{\mathcal{X}}(t)$, and (ii) $\kappa_1(\x[1], t) = 0$ for all $\x \in \bar{\mathcal{X}}(t)$. Then, with proper choice of $K \ge 1$ and $\varDelta \ge 0$, the controller is locally robustness satisfying for $\rho^{\psi}(\x[1](t)) \ge \gamma(t)$ in  $\mathcal{D}$ . 
\end{lemma}
\begin{proof}
	The proof is similar and follows the same steps as that of Theorem \ref{theorem:localSatisfaction}. Let the system at time $\tau$ be at a state $\x(\tau) \in \mathcal{D}$. By definition of local robustness satisfaction, we assume $\rho^{\psi}(\x[1](\tau)) \ge \gamma(\tau)$ holds. Furthermore, we know that $\norm[2]{\v(\x(\tau))} \ge v_{\text{min}}$ as $\x(\tau) \in \mathcal{D}$. Due to the continuity of $\v(\x)$, there exists a closed ball around $\x(\tau)$ for which $\norm[2]{\v(\x)}$ is bounded from below by some $0 < v'_{\text{min}} \le v_{\text{min}}$ by the extreme value theorem. The input $\u[1](\x, t)$ thus satisfies the Lipschitz condition in this ball, as well as the input $\u[2](\x, t)$ by the assumption of the theorem. Thus, the entire system differential equation (\ref{eq:unicycleSystem}) satisfies the conditions of Lemma \ref{lemma:odeSolution}, implying that a unique solution exists within some $[\tau, \tau + \delta]$ time interval. As $\v(\x)$ changes continuously, this $\delta$ value can be chosen small enough such that $\norm[2]{\v(\x)} \ge v'_{\text{min}}$ remains true during the entire duration, which allows us to show $\rho^{\psi}(\x[1](t)) \ge \gamma(t)$ for all $t \in [\tau, \tau + \delta]$ along the same lines as in the previous theorem by invoking Lemma \ref{lemma:nagumo}. Indeed, in case the controller parameters are chosen to satisfy $(K - 1)v'^2_{\text{min}} \ge \varDelta$, for the time derivative of $\rho^{\psi}(\x[1])$ at the crucial $\x \in \lbar{\mathcal{X}}(t)$ states, we have:
	\begingroup
	\allowdisplaybreaks
	\begin{align*}
	\dot{\rho}^{\psi}&(\x[1]) = \dfrac{\partial \rho^{\psi}(\x[1])}{\partial \x[1]} \left(f_1(\x[1]) + \w[1]\right) + \dfrac{\partial \rho^{\psi}(\x[1])}{\partial \x[1]} g_{11}(\x[2]) \u[1] \\
	&=\dfrac{\partial \rho^{\psi}(\x[1])}{\partial \x[1]} \left(f_1(\x[1]) + \w[1]\right) + \v(\x)\tp \dfrac{\kappa_1(\x[1],t) K}{\norm[2]{\v(\x)}^{2} + \varDelta}\v(\x)\\
	&\ge \dfrac{\partial \rho^{\psi}(\x[1])}{\partial \x[1]} \left(f_1(\x[1]) + \w[1]\right) + \dfrac{\kappa_1(\x[1],t)}{\norm[2]{\v(\x)}^{2}}\v(\x)\tp \v(\x) \\
	&= \dfrac{\partial \rho^{\psi}(\x[1])}{\partial \x[1]} \left(f_1(\x[1]) + \w[1]\right) + \kappa_1(\x[1], t) \\
%	&\ge \dfrac{\partial \rho^{\psi}}{\partial \x[1]} f_1 - \max_{\w[1]}\norm{\dfrac{\partial \rho^{\psi}}{\partial \x[1]} \w[1]} \\
%	&\phantom{\ge}+ \dot{\gamma}(t) - \dfrac{\partial \rho^{\psi}}{\partial \x[1]} f_1+ \max_{\w[1]}\norm{\dfrac{\partial \rho^{\psi}}{\partial \x[1]} \w[1]} \\
	&\ge \dot{\gamma}(t)
	\end{align*}
	%	\begin{align*}
	%		\dot{\rho}^{\psi}(\x[1]) &= \dfrac{\partial \rho^{\psi}(\x[1])}{\partial \x[1]} \left(f_1(\x) + \w[1]\right) + \dfrac{\partial \rho^{\psi}(\x[1])}{\partial \x[1]} g_{11}(\x[2]) \u[1] \\
	%		&=\dfrac{\partial \rho^{\psi}(\x[1])}{\partial \x[1]} \left(f_1(\x) + \w[1]\right) + \kappa_1(\x[1], t) \dfrac{\v(\x)\tp \v(\x)}  {\norm{\v(\x)}^{2}} \\
	%		&= \dfrac{\partial \rho^{\psi}(\x[1])}{\partial \x[1]} \left(f_1(\x) + \w[1]\right) + \kappa_1(\x[1], t) \\
	%		&\ge \dfrac{\partial \rho^{\psi}(\x[1])}{\partial \x[1]} f_1(\x) - \max_{\w[1]}\norm{\dfrac{\partial \rho^{\psi}(\x[1])}{\partial \x[1]} \w[1]} \\
	%		&\phantom{\ge}+ \dot{\gamma}(t) - \dfrac{\partial \rho^{\psi}(\x[1])}{\partial \x[1]} f_1(\x) + \max_{\w[1]}\norm{\dfrac{\partial \rho^{\psi}(\x[1])}{\partial \x[1]} \w[1]} \\
	%		&= \dot{\gamma}(t)
	%	\end{align*}
	\endgroup
	as required by the lemma, and the proof is complete.
\end{proof}

Our task is now to choose the control $\u[2]$ to satisfy the augmented robustness specification for $\psi_{\text{aug}}$ within $\mathcal{D}$. The time derivative of the corresponding robustness is given as:
\begin{align*}
\dot{\rho}^{\psi_{\text{aug}}}(\x) &= \dfrac{\partial \rho^{\psi_{\text{aug}}}(\x)}{\partial \x[1]} \x<\dot>[1] + \dfrac{\partial \rho^{\psi_{\text{aug}}}(\x)}{\partial \x[2]} \x<\dot>[2] \\
&= \dfrac{\v(\x)\tp}{\norm[2]{\v(\x)}} \left(\dfrac{\partial \v(\x)}{\partial \x[1]} \x<\dot>[1] + \dfrac{\partial \v(\x)}{\partial \x[2]} \x<\dot>[2]\right).
\end{align*}
After substituting in the dynamics for $\x<\dot>[1]$ and $\x<\dot>[2]$ from (\ref{eq:unicycleSystem}), this expression takes the general form:
\begin{equation*}
\dot{\rho}^{\psi_{\text{aug}}}(\x) = F(\x,\w) + G(\x)\u[1] + \v[\text{aug}](\x)\tp\u[2], 
\end{equation*}
where $F$ is composed of the unknown terms:
\begin{align*}
F(\x,\w) = &\dfrac{\v(\x)\tp}{\norm[2]{\v(\x)}} \left[\dfrac{\partial \v(\x)}{\partial \x[1]} \left(f_1(\x[1]) + \w[1]\right)\right. \\
&\left.{+}\dfrac{\partial \v(\x)}{\partial \x[2]} \left(f_2(\x) + \w[2]\right)\right],
\end{align*}
the coefficient of $\u[1]$ is
\begin{equation} \label{eq:Gdef}
G(\x) = \dfrac{\v(\x)\tp}{\norm[2]{\v(\x)}} \left[\dfrac{\partial \v(\x)}{\partial \x[1]} g_{11}(\x[2]) + \dfrac{\partial \v(\x)}{\partial \x[2]} g_{21}(\x)\right],
\end{equation}
and the coefficient of $\u[2]$ is given as:
\begin{equation} \label{eq:vAugDef}
\v[\text{aug}](\x)\tp = \dfrac{\v(\x)\tp}{\norm[2]{\v(\x)}} \dfrac{\partial \v(\x)}{\partial \x[2]} g_{22}(\x).
\end{equation}
\begin{assumption} \label{assumption:controllabilityu2}
	For the region of interest of the augmented robustness specification, we have:
	\begin{equation}
	\v[\text{aug}](\x) \ne \nvec, \quad \forall \x \in \mathcal{D} : \exists t\ \st \ \x \in \mathcal{X}_{\text{aug}}(t).
	\end{equation}
\end{assumption}

\begin{lemma} \label{lemma:unicycleTheorem}
	Assume that the controller $\u[1](\x, t)$ for $\u[1]$ is Lipschitz continuous in $\x$ and piecewise continuous in $t$ in domain $\mathcal{D}$, and that Assumptions \ref{assumption:general} and \ref{assumption:controllabilityu2} hold. Define the control law for $\u[2]$ as
	\begin{equation} \label{eq:unicycleControlu2}
	\u[2](\x,t) = \begin{cases}
	\nvec \qquad &\text{if } \x \in \mathcal{A}_{\text{aug}}(t), \\
	 \dfrac{\kappa_\text{aug}(\x, t) K_{\text{aug}} \v[\text{aug}](\x)}{\norm{\v[\text{aug}](\x)}^{2} + \varDelta_{\text{aug}}} \qquad &\text{if } \x \notin \mathcal{A}_{\text{aug}}(t),
	\end{cases}
	\end{equation}
	where the coefficient $\kappa_\text{aug}(\x, t) \ge 0$ is continuous and satisfies (i) $\kappa_\text{aug}(\x, t) \ge \dot{\gamma}_\text{aug}(t) - G(\x)\u[1] + B_\text{aug}(\x)$ with $B_\text{aug}(\x) \ge \max_{\w} \norm[2]{F(\x, \w)}$ for all $\x \in \lbar{\mathcal{X}}_{\text{aug}}(t)$, and (ii) $\kappa_\text{aug}(\x, t) = 0$ for all $\x \in \bar{\mathcal{X}}_{\text{aug}}(t)$. Then, by properly selecting $K_{\text{aug}} \ge 1$ and $\varDelta_{\text{aug}} \ge 0$, the control law (\ref{eq:unicycleControlu2}) achieves local robustness satisfaction of the augmented robustness specification $\rho^{\psi_{\text{aug}}}(\x(t)) \ge \gamma_{\text{aug}}(t)$ in the domain $\mathcal{D}$.
\end{lemma}
\begin{proof}
	The proof again follows exactly the same lines as that of Lemma \ref{lemma:phiSatisfaction}, first showing that Lemma \ref{lemma:odeSolution} is applicable within $\mathcal{D}$ and guarantees the existence of a unique solution for some period of time. Then, Lemma \ref{lemma:nagumo} is used to show that with $(K_{\text{aug}} - 1)v^2_{\text{aug,min}} \ge \varDelta_{\text{aug}}$ we will always have $\dot{\rho}^{\psi_{\text{aug}}}(\x(t)) \ge \dot{\gamma}_{\text{aug}}(t)$ at $\x \in \lbar{\mathcal{X}}_{\text{aug}}(t)$ due to the structure of the introduced control law for $\u[2]$, which in turn implies the desired local robustness satisfaction.
\end{proof}
This leads us to the main result of this section.
\begin{theorem} \label{theorem:unicycleTheorem}
	Let Assumptions \ref{assumption:general} and \ref{assumption:controllabilityu2} hold. Then, the control laws (\ref{eq:unicycleControlu1}) and (\ref{eq:unicycleControlu2}) together achieve local robustness satisfaction of the conjoined specification $\rho^{\psi}(\x[1](t)) \ge \gamma(t)$ and $\rho^{\psi_{\text{aug}}}(\x(t)) \ge \gamma_{\text{aug}}(t)$ within the domain $\mathcal{D}$.
\end{theorem}
\begin{proof}
	Let $\x(\tau) \in \mathcal{D}$ be the state at time $\tau$ for which both specifications $\rho^{\psi}(\x(\tau)) \ge \gamma(\tau)$ and $\rho^{\psi_{\text{aug}}}(\x(\tau)) \ge \gamma_{\text{aug}}(\tau)$ are satisfied. Lemmas \ref{lemma:phiSatisfaction} and \ref{lemma:unicycleTheorem} individually guarantee the existence of a unique solution for finite times $[\tau, \tau + \delta_1]$ and $[\tau, \tau + \delta_2]$, with $\delta_1 > 0$ and $\delta_2 > 0$. The lemmas also guarantee local robustness satisfaction during this period for the two tasks, independently of one another. This implies that during the finite time interval $t \in [\tau, \tau + \delta]$, where $\delta = \min(\delta_1, \delta_2) > 0$, a unique solution exists and both specifications remain satisfied, as desired.
\end{proof}

\begin{remark}
	The state of the system is guaranteed to remain in $\mathcal{D}$ for any period of time for which a solution exists. This is readily seen as the controller is locally robustness satisfying for the temporal behaviors of $\psi$ and $\psi_\text{aug}$, which implies the state must remain in $\mathcal{D}$ due to its definition. The results for global task satisfaction from Theorem \ref{theorem:globalSatisfaction} and Corollary \ref{corollary:disjointConjunctions} using the obtained controller $\u$ thus continue to hold as the controller remains well-defined throughout time.
\end{remark}

%To ease understanding of the quantities introduced throughout the derivations of this section, we continue the study of the sample unicycle navigation task.
\begin{example}[Unicycle navigation task - continued]
	The redefined term (\ref{eq:vDef2}) for $\v(\x)$ becomes $\v(\x) = -\norm[2]{\e[g]}^{-1}\left(\e[g]\tp \n\right)$, where the unicycle faces the $\n = [\cos\theta \ \sin\theta]\tp$ direction. The robustness measure for the augmented task $\phi_{\text{aug}}$ is given by $\rho^{\psi_{\text{aug}}}(\x) = \norm[2]{\v(\x)} - v_{\text{min}}$ accordingly. The coefficient $G(\x)$ in the time derivative of this term, given by (\ref{eq:Gdef}), becomes
	$G(\x) = -\norm[2]{\e[g]}^{-1}\norm[2]{\v(\x)}^{-1}\v(\x)\tp\left(1 - \norm[2]{\v(\x)}^{2}\right)$. The coefficient $\v[\text{aug}](\x)$, expressed in (\ref{eq:vAugDef}), takes the form:
	\begin{equation*}
	\v[\text{aug}](\x) = -\dfrac{1}{\norm[2]{\e[g]}} \dfrac{\v(\x)\tp}{\norm[2]{\v(\x)}} \e[g]\tp \n[\perp],
	\end{equation*}
	where $\n[\perp] = [-\sin\theta \ \cos\theta]\tp$ is perpendicular to the unicycle's direction.
	The terms $\v$ and $\v[\text{aug}]$ both become zero when the unicycle is perpendicular to the target error $\e[g]$. This case is excluded from the set $\mathcal{D}$ as $\rho^{\psi_{\text{aug}}} = -v_{\text{min}} < 0 = \gamma_{\text{aug}}(t)$ for such a case. The term $\v[\text{aug}]$ also becomes zero when the unicycle is parallel to the target error. In such a case, $\rho^{\psi_{\text{aug}}} = 1 - v_{\text{min}}$ and this can also be excluded from the region of interest $\mathcal{X}_{\text{aug}}(t)$ by an appropriate choice of $\alpha = \varGamma_{\text{aug}}(t) < 1-v_{\text{min}}$, ensuring Assumption \ref{assumption:controllabilityu2} is satisfied. When $\e[g] = \nvec$, the terms become ill-defined due to the divisions by $\norm[2]{\e[g]}$. Consider, on the other hand, the task of avoiding an obstacle. Then, the domain $\mathcal{D}$ does not contain the point where $\e[g] = \nvec$ as the obstacle should be avoided, and thus $\u[2](\x,t)$ is well-defined in $\mathcal{D}$ and the results of Theorem \ref{theorem:unicycleTheorem} for robustness satisfaction hold.
	
	From a practical point of view, the controller can also be used for reaching a target location as $\e[g] = \nvec$ is a measure zero set. (Theoretically, it should be combined with an arbitrarily small radius target avoidance to have global guarantees of task satisfaction). Sample trajectories for solving the \ac{STL} task outlined in Example \ref{example:1}, with $\kappa_1 = 2e^{-\frac{\rho^{\psi}(\x_1) - \gamma(t)}{\varGamma(t) - \rho^{\psi}(\x_1)}}$ and $\kappa_2 = (-G(\x)\u[1] + 20)e^{-\frac{\rho^{\psi_{\text{aug}}}(\x) - \gamma(t)}{\varGamma(t) - \rho^{\psi_{\text{aug}}}(\x)}}$ defining the controls (\ref{eq:unicycleControlu1}) and (\ref{eq:unicycleControlu2}), are shown in Fig. \ref{fig:unicycleSample}. The $K$ and $\varDelta$ parameters of the two controllers for $\u[1]$ and $\u[2]$ were set to 1 and 0. Process noise with covariance $\diag(0.5,\ 0.5,\ 5)$ was added to the system as a disturbance. The evolution of the robustness metrics is shown in Fig. \ref{fig:unicycleRobustness}.
\end{example}

\begin{figure}[h]
	\centering
	\includegraphics[width=.53\linewidth]{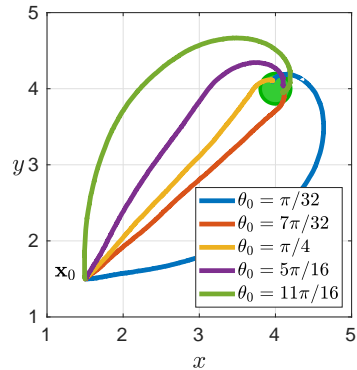}
	\caption{Sample trajectories for the unicycle navigation task example from various initial unicycle angles $\theta_0$.}
	\label{fig:unicycleSample}
\end{figure}
\vspace{-4mm}
\begin{figure}[h]
	\centering
	\begin{subfigure}[b]{0.23\textwidth}
		\centering
		\includegraphics[width=\textwidth]{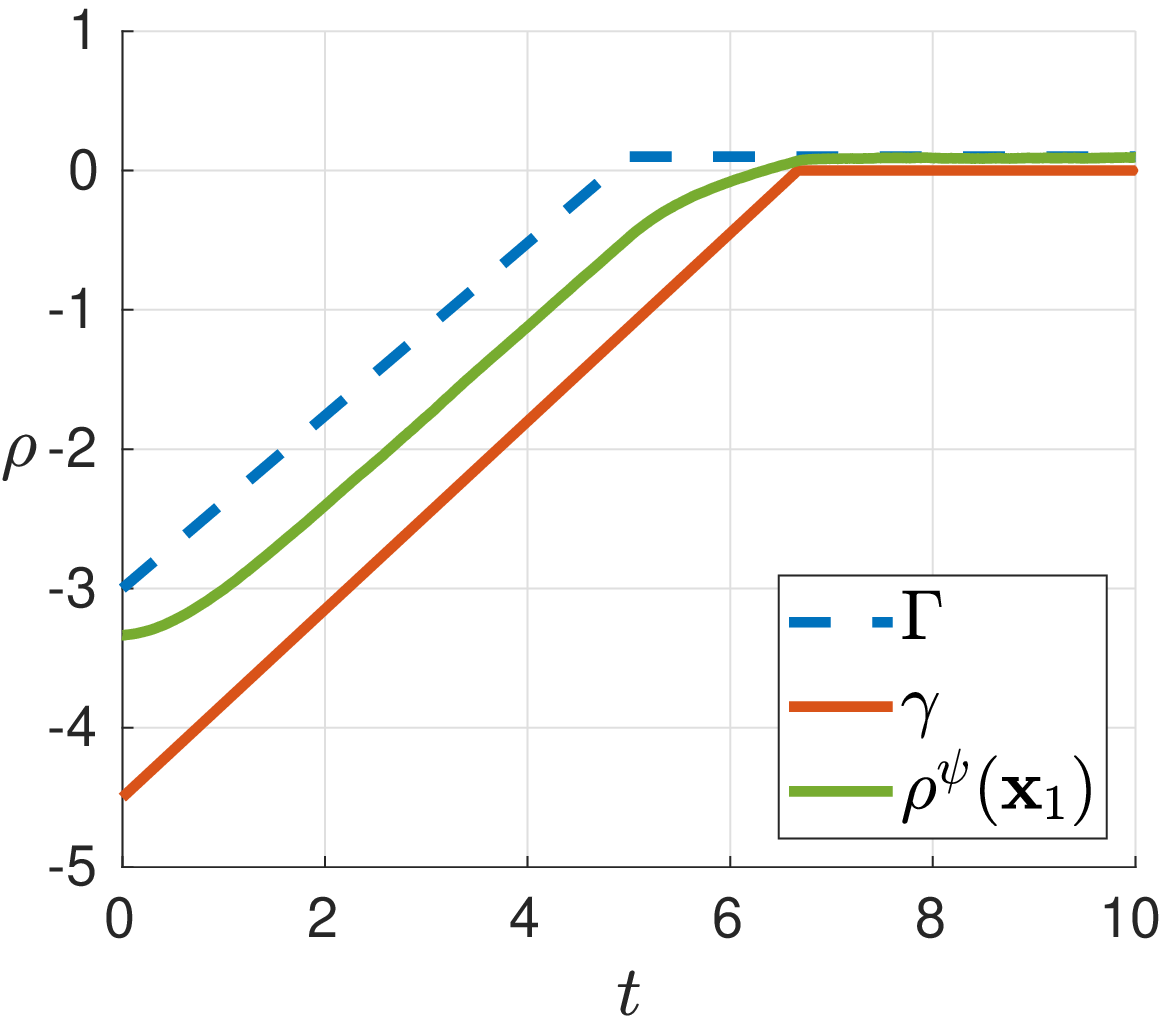}
	\end{subfigure}
	\begin{subfigure}[b]{0.23\textwidth}
		\centering
		\includegraphics[width=\textwidth]{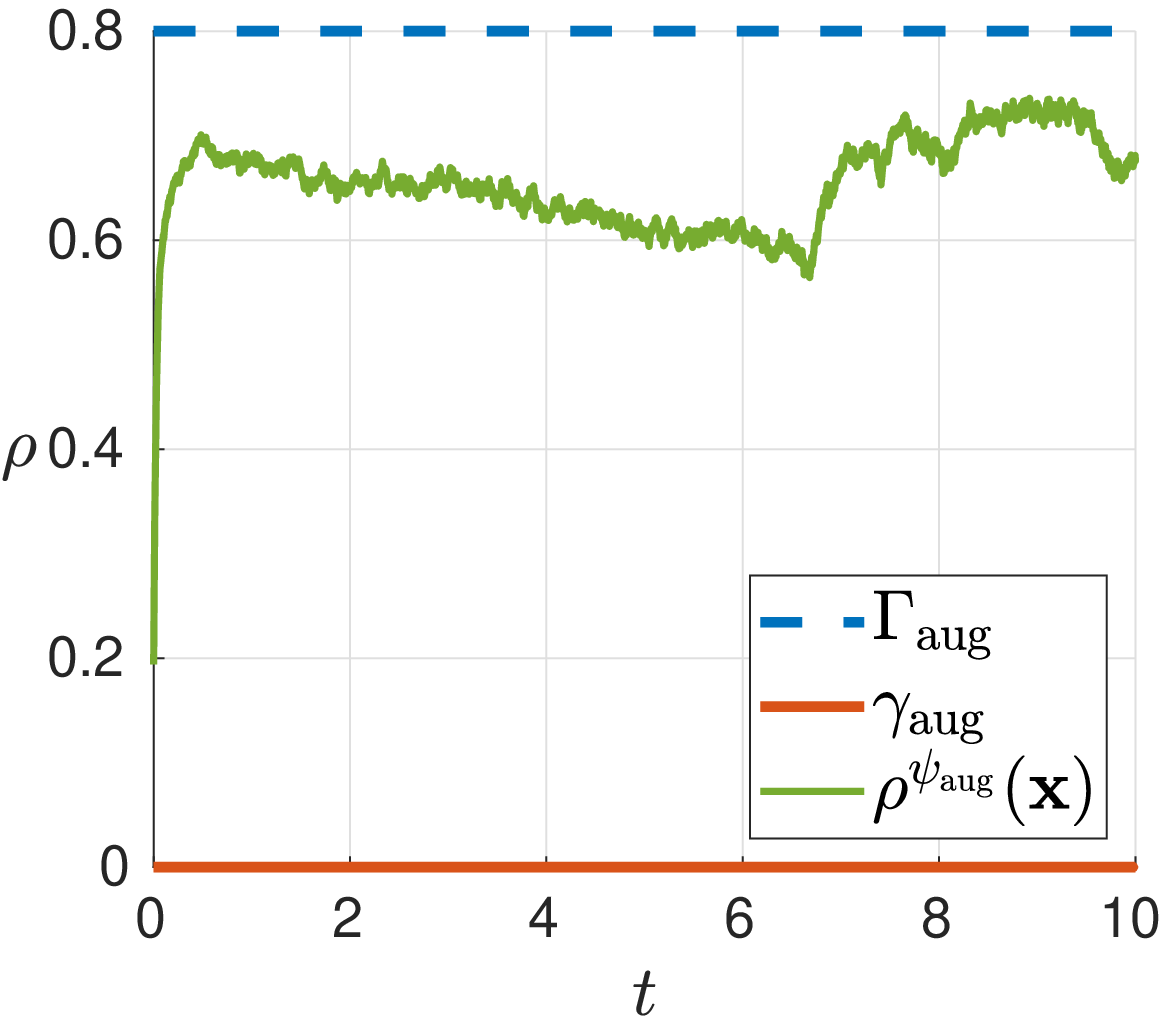}
	\end{subfigure}
	\caption{Evolution of robustness measures $\rho^{\psi}$ and $\rho^{\psi_\text{aug}}$ for the unicycle example in the sample case $\theta_0 = 11\pi/16$.}
	\label{fig:unicycleRobustness}
\end{figure}

%%%%%%%%%%%%%%%%%%%%%%%%%%%%%%%%%%%%%%%%%%%%%%%%%%%%%%%%%%%%%%%%%%%%%%%%%%%%%%%%
\section{Combining controllers} \label{section:practical}

In this section, we examine the possibility and practicalities associated with using the derived controllers in combination with one another in order to extend the range of \ac{STL} task specifications we can satisfy. The motivation behind this is that controllers for a single robustness specification - \textit{elementary controllers} - are simple and inexpensive to calculate, and even though robustness satisfaction guarantees are lost by combining them, the result can still serve as a good guiding controller for solving tasks, as shown in \cite{varnai2019prescribedARXIV}. 

We propose an approach for combining elementary controllers for the generalized unicycle system. Practical considerations are also given to highlight some aspects of combining elementary controllers and to provide initial insight into a more in-depth study of this topic for future work. The take-aways are illustrated using a simple navigation task.

\subsection{Combining elementary controllers}

Consider a conjunction of $M$ specifications $\rho^{\psi_{(i)}}(\x[1](t)) \ge \gamma_{(i)}(t)$ for $i=1\dots M$; for all quantities, the subscript $(i)$ indicates association to the $i$-th specification. We assume these are elementary in the sense that each admits a locally robustness satisfying controller defined by $\u[1,(i)](\x,t)$ and $\u[2,(i)](\x,t)$. The individual controls $\u[1,(i)]$ can be intuitively combined by taking weighted average:
\begin{equation} \label{eq:combinedu1}
\u[1] := \dfrac{\sum_{i = 1}^{M} \alpha_{(i)} \u[1,(i)]}{\sum_{i = 1}^{M}\alpha_{(i)}}
\end{equation}
to determine a consensus for $\u[1]$, which directly influences the evolution of the different robustness metrics $\rho^{\psi_{(i)}}(\x[1](t))$. The weights are chosen such that higher priority is given as $\x \rightarrow \lbar{\mathcal{X}}_{(i)}(t)$, e.g., with $\alpha_{(i)} = \frac{\varGamma_{(i)}(t) - \rho^{\psi_{(i)}}(\x(t))}{\varGamma_{(i)}(t) - \gamma_{(i)}(t)}$ if $\rho^{\psi_{(i)}}(\x(t)) \le \varGamma_{(i)}(t)$ and $\alpha_{(i)} = 0$ otherwise. A similar scheme is then employed for deciding the input $\u[2]$, i.e., $\u[2] = \frac{\sum_{i = 1}^{M} \alpha_{(i)} \u[2,(i)]}{\sum_{i = 1}^{M}\alpha_{(i)}}$. The weights again serve to mainly exert control action from the elementary controller $i$ whose respective $\rho^{\psi_{(i)}}$ robustness measure is the most violating.

\subsection{Practical considerations}

An elementary controller gives satisfaction guarantees if its $\v_{(i)}(\x)$ term in (\ref{eq:vDef2}) remains non-zero. However, with a conjunction of $M$ specifications, this requirement might be overly restrictive to allow for feasible trajectories. For example, it might be physically impossible for a unicycle to pass by a circular obstacle without becoming perpendicular to it. Elementary controllers aiming to avoid such configurations might thus be working \textit{against} an actual feasible trajectory!

For simplicity, consider a single robustness specification for some formula $\psi$, easing the notation to drop the $(i)$ subscripts. If the input $\u[2]$ is not used to keep $\v(\x) \ne \nvec$, a natural idea is to use it to increase the robustness metric $\rho^{\psi}$ instead. Namely, $\u[2]$ appears in the second derivative of $\rho^{\psi}(\x[1])$, and could potentially be used to push the system towards increasing $\rho^{\psi}$.

It is instructive to examine how the second derivative of $\rho^{\psi}(\x[1])$ depends on the input $\u[2]$ under the derived control law (\ref{eq:unicycleControlu1}) for $\u[1]$. Towards this end, let us first rewrite the expression (\ref{eq:unicycledRho}) for the time derivative of $\rho^{\psi}(\x[1])$ in the form:
\begin{equation}
\dot{\rho}^{\psi}(\x[1]) = \dot{\rho}_{fw}^{\psi}(\x[1],\w[1]) + \v(\x)\tp \u[1](\x, t),
\end{equation}
where the introduced $\dot{\rho}_{fw}^{\psi}(\x[1],\w[1]) = \frac{\partial \rho^{\psi}(\x[1])}{\partial \x[1]} (f_{1}(\x[1]) + \w[1])$. The second derivative is then given by:
\begin{equation}
\ddot{\rho}^{\psi}(\x[1]) = \ddot{\rho}_{fw}^{\psi}(\x[1],\w[1]) + \u[1]\tp \v<\dot>(\x) + \v(\x)\tp \u<\dot>[1](\x, t).
\end{equation}
The second input $\u[2]$ will only appear in the last two terms as part of $\x<\dot>[2]$ when the time derivatives of $v(\x)$ and $\u[1](\x,t)$ are taken. For the middle term $\u[1]\tp \v<\dot>(\x)$, we have:
\begin{equation} \label{eq:ddrhop1}
\u[1]\tp \v<\dot>(\x)  = \u[1]\tp\left(\dfrac{\partial \v(\x)}{\partial \x[1]} \x<\dot>[1] + \dfrac{\partial \v(\x)}{\partial \x[2]} \x<\dot>[2]\right).
\end{equation}
For the last term $\v(\x)\tp \u<\dot>[1](\x, t)$, by inserting the control law (\ref{eq:unicycleControlu1}) for $\u[1](\x,t)$ one obtains:
\begin{align}
\v\tp\dfrac{\partial \u[1]}{\partial \x[2]} \x<\dot>[2] &= \v\tp \left[\kappa_1 K\v \dfrac{\partial (\v\tp\v + \varDelta)^{-1}}{\partial \v} + \dfrac{\kappa_1 K\I}{\v\tp \v + \varDelta}\right] \dfrac{\partial \v}{\partial \x[2]} \x<\dot>[2] \notag \\ 
&= \v\tp \left[\dfrac{-2 \kappa_1 K \v \v\tp  + \kappa_1 K(\v\tp \v + \varDelta)\I}{(\v\tp \v + \varDelta)^2}\right] \dfrac{\partial \v}{\partial \x[2]} \x<\dot>[2] \notag \\
&= -\left(\dfrac{\v\tp\v - \varDelta}{\v\tp \v + \varDelta}\right) \u[1]\tp \dfrac{\partial \v}{\partial \x[2]} \x<\dot>[2] \label{eq:ddrhop2}
\end{align}
after some simplifications. The arguments of each term have been dropped for better readability.
%\begin{align}
%\v\tp\dfrac{\partial \u[1]}{\partial \x[2]} \x<\dot>[2] &= \v\tp \left[\kappa_1\v \dfrac{\partial (\v\tp\v + \varDelta)^{-1}}{\partial \v} + \dfrac{\kappa_1\I}{\v\tp \v + \varDelta}\right] \dfrac{\partial \v(\x)}{\partial \x[2]} \x<\dot>[2] \notag \\ 
%&= \v\tp \left[\dfrac{-2 \kappa_1 \v \v\tp  + \kappa_1(\v\tp \v + \varDelta)\I}{(\v\tp \v + \varDelta)^2}\right] \dfrac{\partial \v(\x)}{\partial \x[2]} \x<\dot>[2] \notag \\
%&= \left[\dfrac{-\kappa_1\v\tp\v\v\tp + \kappa_1\v\tp\varDelta}{(\v\tp \v + \varDelta)^2}\right] \dfrac{\partial \v(\x)}{\partial \x[2]} \x<\dot>[2] \notag \\
%&= \left[\dfrac{-\kappa_1(\v\tp\v - \varDelta)\v\tp}{(\v\tp \v + \varDelta)^2}\right] \dfrac{\partial \v(\x)}{\partial \x[2]} \x<\dot>[2] \notag \\
%&= -\left(\dfrac{\v\tp\v - \varDelta}{\v\tp \v + \varDelta}\right) \u[1]\tp \dfrac{\partial \v(\x)}{\partial \x[2]} \x<\dot>[2] \label{eq:ddrhop2}
%\end{align}
Adding the contribution of terms involving $\x<\dot>[2]$ (and hence $\u[2]$ after substituting in the system dynamics) from (\ref{eq:ddrhop1}) and (\ref{eq:ddrhop2}), we have that the component of $\ddot{\rho}^{\psi}(\x[1])$ depending on $\x<\dot>[2]$ is:
\begin{align*}
\ddot{\rho}^{\psi}_{\x[2]}(\x[1]) &:= \u[1]\tp \dfrac{\partial \v(\x)}{\partial \x[2]} \x<\dot>[2] - \left(\dfrac{\norm[2]{\v(\x)}^2 - \varDelta}{\norm[2]{\v(\x)}^2 + \varDelta}\right) \u[1]\tp \dfrac{\partial \v(\x)}{\partial \x[2]} \x<\dot>[2] \\
&= \dfrac{2\varDelta}{\norm[2]{\v(\x)}^2 + \varDelta}\u[1]\tp \dfrac{\partial \v(\x)}{\partial \x[2]} \x<\dot>[2].
\end{align*}
Substituting in the dynamics (\ref{eq:unicycleSystem}) for $\x<\dot>[2]$, the dependency on $\u[2]$ can be finally seen to be:
\begin{equation*}
\ddot{\rho}^{\psi}_{\u[2]}(\x[1]) = \dfrac{2\varDelta}{\norm[2]{\v(\x)}^2 + \varDelta}\u[1]\tp \dfrac{\partial \v(\x)}{\partial \x[2]} g_{22}(\x) \u[2] := \v[2](\x, \u[1])\tp \u[2]
\end{equation*}
If the controller for $\u[1]$ employs no regularization and so $\varDelta = 0$, then $\u[2]$ does not have an effect on the evolution of this term. This is expected, because $\u[1]$ from (\ref{eq:unicycleControlu1}) normalizes the term $\v(\x)$ when $\varDelta = 0$, effectively removing its influence on the change of the robustness metric. To allow this normalization, $\v(\x)$ must be kept nonzero using $\u[2]$.

As discussed, however, the individual $\v(\x)$ terms may become zero when combining different elementary robustness specifications. Therefore, regularization is needed to have well-defined control signals in such configurations and we must have $\varDelta \ne 0$. With this choice, $\u[2]$ has an impact on each $\ddot{\rho}^{\psi_{(i)}}_{\u[2]}(\x[1])$, and it is intuitively beneficial to use it to increase this term as $\psi_{(i)}$ nears violation, i.e., as $\x \rightarrow \lbar{\mathcal{X}}_{(i)}(t)$. In accordance with the previous controllers, we can thus define a more practical law for each specification in general as:
\begin{equation} \label{eq:unicyclePracticalu2}
\u<\tilde>[2](\x,t) = \begin{cases}
\nvec \qquad &\text{if } \x \in \mathcal{A}(t), \\
 \dfrac{\kappa_2(\x, t) K_2 \v[2](\x, \u[1])}{\norm{\v[2](\x, \u[1])}^{2} + \varDelta_{2}} \qquad &\text{if } \x \notin \mathcal{A}(t),
\end{cases}
\end{equation}
where $K_2 \ge 1$, $\varDelta_2 \ge 0$, and $\kappa_2$ is chosen similarly as before to increase as $\x \rightarrow \lbar{\mathcal{X}}(t)$ and become zero as $\x \rightarrow \bar{\mathcal{X}}(t)$.

Note that, as opposed to the controller (\ref{eq:unicycleControlu2}), $\u<\tilde>[2](\x,t)$ depends on $\u[1]$. When combining controllers, the consensus (\ref{eq:combinedu1}) is thus used to determine the elementary controls that are then averaged for $\u[2]$. For example, if a unicycle has been forced to go towards an obstacle, this will be taken into account while computing the controller $\u[2,(i)]$ whose aim is to avoid the obstacle, and $\u[2,(i)]$ will now attempt to turn the unicycle away from it as illustrated in the following section.

\subsection{Case study}

Consider the unicycle navigation task of reaching $r_g = 0.2$ distance within a goal region at $\x[g] = [1.0\ 3.5]\tp$ while avoiding a circular obstacle with radius $r_o = 1.2$ located at $\x[o] = [2.5\ 2.0]\tp$. The task is given by $\phi = \eventually[0][10]\psi_{(1)} \and G\psi_{(2)}$, where $\psi_{(1)} = \left\{r_g - \norm[2]{\x[1] - \x[g]} \ge 0 \right\}$ and $\psi_{(2)} = \left\{\norm[2]{\x[1] - \x[o]} - r_o \ge 0 \right\}$. The initial state of the unicycle is $\x[1,0] = [3.5\ 0.3]\tp$ and $\x[2,0] = 15\pi/16$. The inputs are constrained as $\norm[2]{\u[1]} = |v| \le 1$ and $\norm[2]{\u[2]} = |\omega| \le 5$.

The \ac{STL} formula $\phi$ can be satisfied by placing constraints on the robustness measures of $\psi_{(1)}$ and $\psi_{(2)}$. For the \textit{eventually} subtask of reaching the goal within 10 seconds, we use $\gamma_{(1)}(t) = -4+2.5t$ and $\varGamma_{(1)}(t) = \min(0.99\cdot r_g, \gamma_{(1)}(t) + 1)$, while for the \textit{always} subtask of avoiding the obstacle, we simply use $\gamma_{(2)}(t) = 0$ and $\varGamma_{(2)}(t) = 0.5$ to achieve this satisfaction. In all elementary controllers, the parameters are set as $K = 1$ and the regularization $\varDelta = 0.5$. The control actions $\u[1,(i)]$ are calculated according to the gains $\kappa_{1,(i)} = 5 \exp\left(-\frac{\rho^{\psi_{(i)}}(\x(t)) - \gamma_{(i)}(t)}{\varGamma_{(i)}(t)- \rho^{\psi_{(i)}}(\x(t))}\right)$ and are then combined according to (\ref{eq:combinedu1}) to determine the velocity $v = \u[1]$.

We compare the performance when combining the two derived elementary controllers for $\u[2]$. The first, defined in (\ref{eq:unicycleControlu2}), gives satisfaction guarantees individually for the two robustness specifications and is referred to as the augmented (`aug') controller. The second, defined in (\ref{eq:unicyclePracticalu2}), takes the discussed practical considerations into account and is labeled as practical (`prac'). For the augmented controller, we define $\rho^{\psi_{\text{aug}, (i)}}(\x) = \norm[2]{\v_{(i)}(\x)} - v_{\text{min}}$ with $v_{\text{min}} = 0.001$. The controller coefficients are $\kappa_{\text{aug},(i)} = (-G \u[1] + 5)\cdot \exp\left(-\frac{\rho^{\psi_{\text{aug}, (i)}}(\x(t)) - \gamma_{\text{aug}, (i)}(t)}{\varGamma_{\text{aug}, (i)}(t)- \rho^{\psi_{\text{aug}, (i)}}(\x(t))}\right)$. For the practical controller, we use the gain $\kappa_{2,(i)} = 20 \exp\left(-\frac{\rho^{\psi_{(i)}}(\x(t)) - \gamma_{(i)}(t)}{\varGamma_{(i)}(t)- \rho^{\psi_{(i)}}(\x(t))}\right)$. 

\begin{figure}[b]
	\centering \vspace{-3mm}
	\begin{subfigure}[b]{.49\linewidth}
		\centering
		\includegraphics[width=\linewidth]{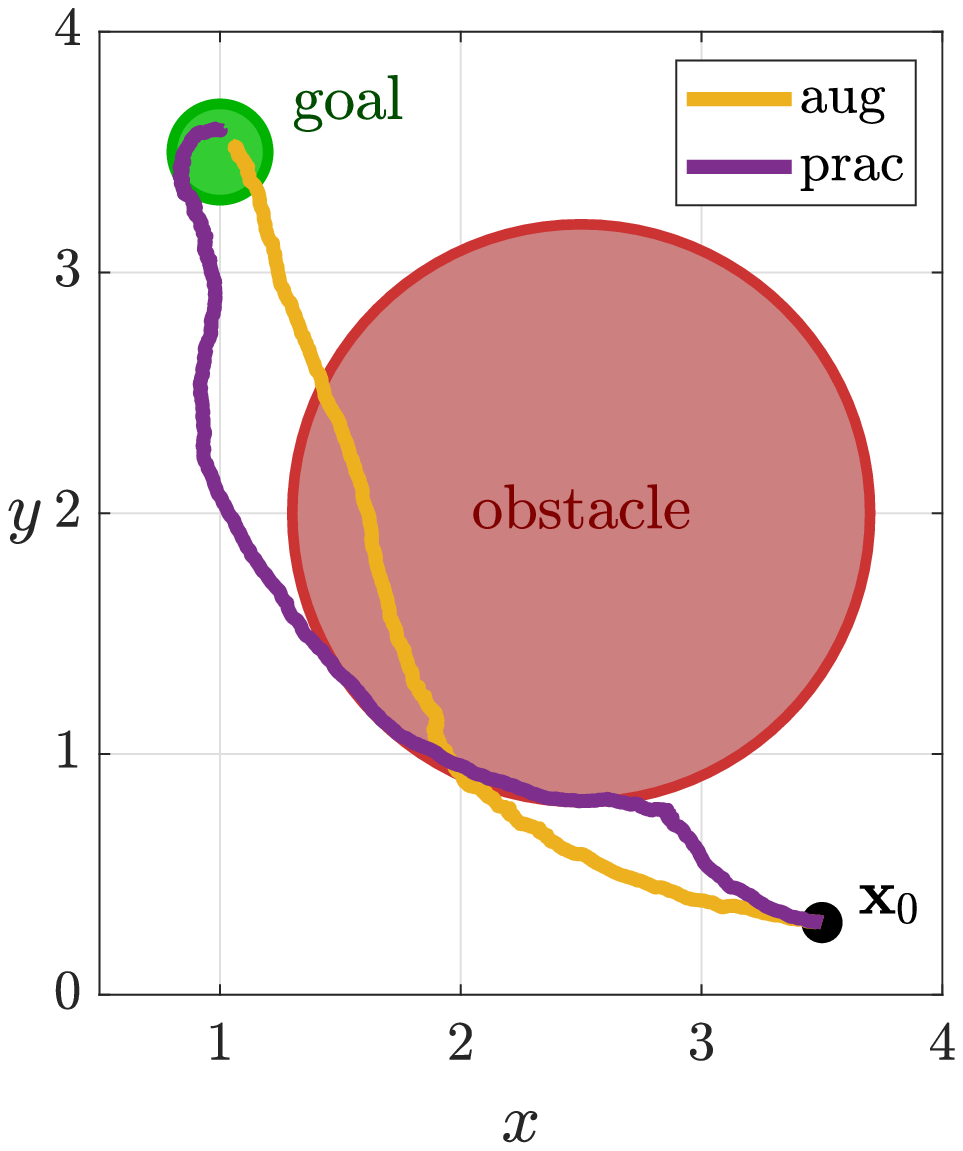}
		\caption{}
	\end{subfigure}\hfill
	\begin{subfigure}[b]{.46\linewidth}
		\centering
		\includegraphics[width=\linewidth,trim={0 8mm 0 0},clip]{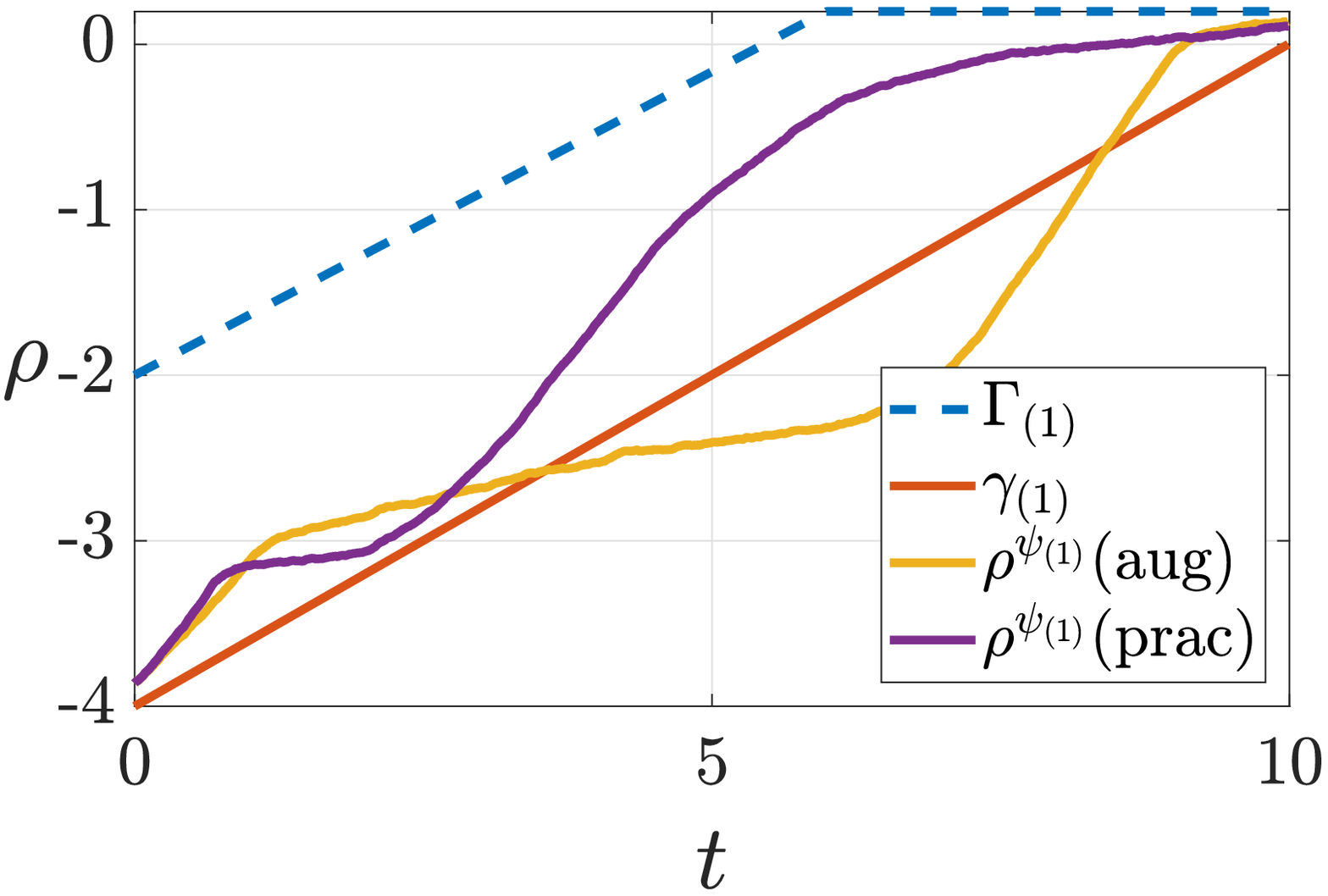}
		\includegraphics[width=\linewidth]{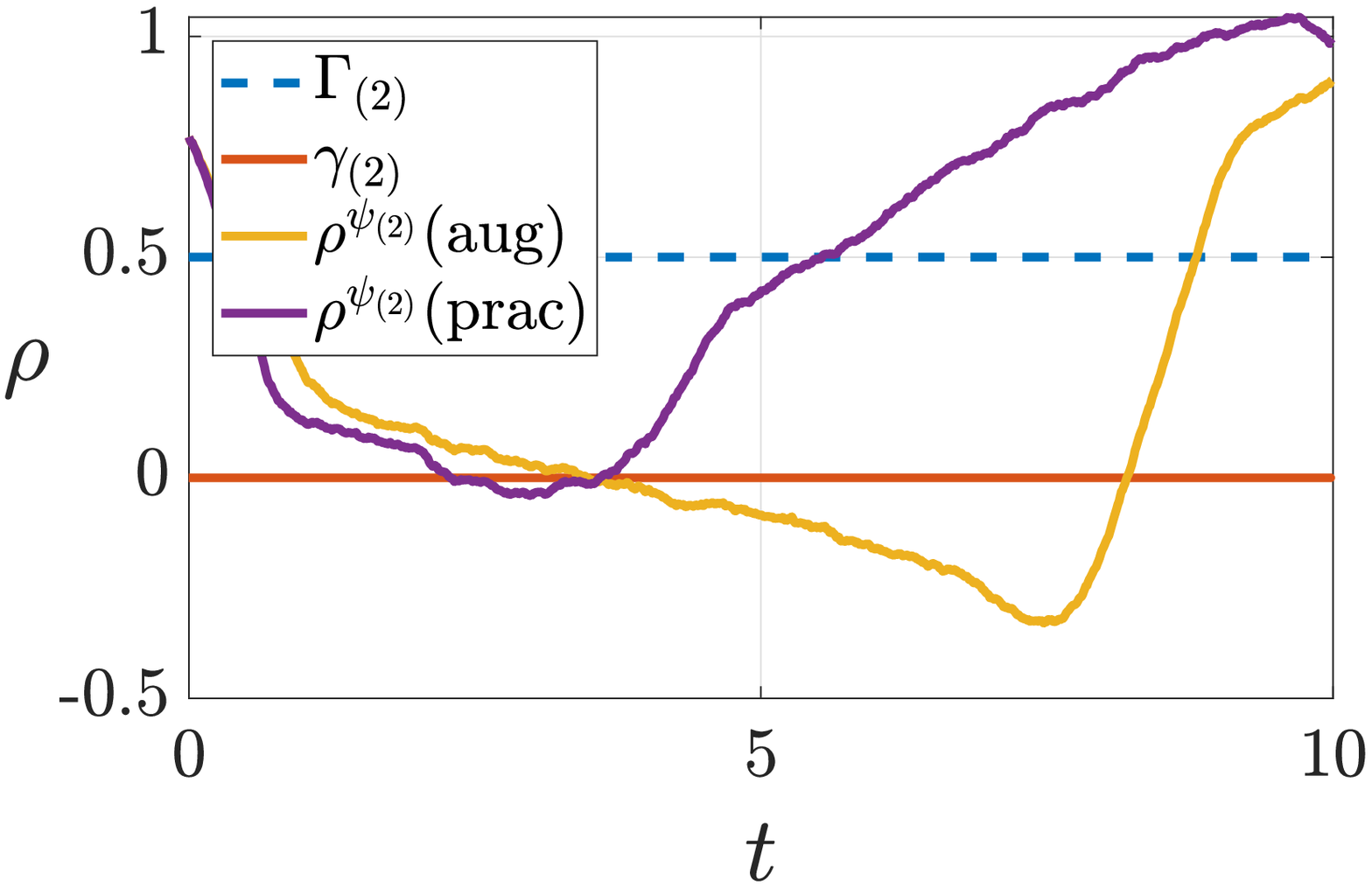}
		\caption{}
	\end{subfigure}
	\caption{(a) Sample trajectories and (b) evolution of robustness measures obtained for the case study navigational task using the `aug' and `prac' controllers (\ref{eq:unicycleControlu2}) and (\ref{eq:unicyclePracticalu2}), respectively.}
	\label{fig:sampleAvoidance}
\end{figure}

A sample result with added process noise is shown in Figure \ref{fig:sampleAvoidance} below. The `aug' controller has trouble avoiding the obstacle as it aims to keep the unicycle oriented towards it, while the specification of reaching the goal region forces the unicycle to still go in that direction. The `prac' controller takes this heading direction into account and steers away from the obstacle instead, almost satisfying the robustness specifications for $\psi_{(1)}$ and $\psi_{(2)}$. The practical controller already gives more effective results with minimal tuning in this simple example, and is expected to aid exploration better in learning algorithms such as in \cite{varnai2019prescribedARXIV}.

%%%%%%%%%%%%%%%%%%%%%%%%%%%%%%%%%%%%%%%%%%%%%%%%%%%%%%%%%%%%%%%%%%%%%%%%%%%%%%%%

\section{Conclusions} \label{section:conclusions}

In this paper, we presented a framework to study the design of gradient-based controllers for dynamical system subject to \ac{STL} task specifications. A class of controllers that give satisfaction guarantees for simple dynamical systems and tasks was introduced. The use of the developed framework was exemplified by deriving controllers for unicycle-like systems as well. Finally, an initial approach on how such elementary controllers can be combined to solve more elaborate task specifications was discussed, and the significance of the related practicalities was highlighted by a unicycle navigation task. The introduced framework and concepts pave way for designing such inexpensive controllers for an even wider range of system dynamics, with their intended use being to effectively aid exploration in learning algorithms.

%%%%%%%%%%%%%%%%%%%%%%%%%%%%%%%%%%%%%%%%%%%%%%%%%%%%%%%%%%%%%%%%%%%%%%%%%%%%%%%%
%% END OF OWN INCLUDES

\addtolength{\textheight}{-12cm}   % This command serves to balance the column lengths

%%%%%%%%%%%%%%%%%%%%%%%%%%%%%%%%%%%%%%%%%%%%%%%%%%%%%%%%%%%%%%%%%%%%%%%%%%%%%%%%

%%%%%%%%%%%%%%%%%%%%%%%%%%%%%%%%%%%%%%%%%%%%%%%%%%%%%%%%%%%%%%%%%%%%%%%%%%%%%%%%

%%%%%%%%%%%%%%%%%%%%%%%%%%%%%%%%%%%%%%%%%%%%%%%%%%%%%%%%%%%%%%%%%%%%%%%%%%%%%%%%
%\section*{APPENDIX}

%\section*{ACKNOWLEDGMENT}

%%%%%%%%%%%%%%%%%%%%%%%%%%%%%%%%%%%%%%%%%%%%%%%%%%%%%%%%%%%%%%%%%%%%%%%%%%%%%%%%

\bibliographystyle{IEEEtran}
\bibliography{IEEEabrv,MyBib}

\end{document}